\documentclass[12pt,a4paper]{article}
\usepackage[centertags]{amsmath}
\usepackage{amsfonts,amsthm,amssymb}
\usepackage{amssymb}
\usepackage{amsmath}
\usepackage{graphicx}

\theoremstyle{definition}
\newtheorem{theorem}{Theorem}
\newtheorem{corollary}{Corollary}

\newtheorem{proposition}{Proposition}
\newtheorem{remark}{Remark}

\begin{document}

\title{Small Support Approximate Equilibria in Large Games}
\author{Yakov Babichenko\footnote{Center for the Mathematics of Information, Department of Computing and Mathematical Sciences, California Institute of Technology. e-mail:babich@caltech.edu.}\footnote{The author wishes to thank Constantinos Daskalakis and Siddarth Barman for useful discussions and comments. The author  gratefully acknowledges support from a Walter S. Baer and Jeri Weiss fellowship.}}

\maketitle

\begin{abstract}
In this note we provide a new proof for the results of Lipton et al. \cite{lipton} on the existence of an approximate Nash equilibrium with logarithmic support size. Besides its simplicity, the new proof leads to the following contributions:

1. For $n$-player games, we improve the bound on the size of the support of an approximate Nash equilibrium.

2. We generalize the result of Daskalakis and Papadimitriou \cite{daskalakis} on small probability games from the two-player case to the general $n$-player case.

3. We provide a logarithmic bound on the size of the support of an approximate Nash equilibrium in the case of graphical games.

\end{abstract}

\section{Introduction}

The problem of the existence of a small-support approximate equilibrium (i.e., every player randomizes among \emph{small} set of his actions) has been studied in the literature for the past two decades. Althofer \cite{althofer} considered two-player zero-sum games and showed existence of approximately optimal strategies with support of size $O(\log m)$, where $m$ is the number of actions. Lipton, Markakis, and Mehta \cite{lipton} later generalized this result to all two-player games; i.e., they showed existence of an approximate equilibrium with support of size $O(\log m)$. This result yields an exhaustive search algorithm for computing an approximate Nash equilibrium with a quasi-polynomial running time ($m^{\log m}$). This is the best-known bound today for computing an approximate Nash equilibrium. Daskalakis and Papadimitriou \cite{daskalakis} generalized the technique of Lipton et al. \cite{lipton} to prove that in two-player games an approximate Nash equilibrium can be computed in polynomial time in games that possess a small-probabilities Nash equilibrium (see definition in Section \ref{sec:small}).

The related problem of the existence of a pure Nash equilibrium (an equilibrium with the minimal support) in subclasses of games has been studied in the literature for much longer; see, e.g., Rosenthal \cite{rosenthal} and Shmeidler \cite{shmid}. A recent paper by Azrieli and Shmaya \cite{shmaya} analyzes the relation between the influence that a player has on the payoffs of other players and the existence of an approximate Nash equilibrium. They show that if the influence is small enough, then such a game has an approximate pure Nash equilibrium.

In this note we provide a new proof for the results of Lipton et al. \cite{lipton} and Daskalakis and Papadimitriou \cite{daskalakis} using similar techniques to those developed by Azrieli and Shmaya \cite{shmaya}. Besides its simplicity, the new proof leads to the following contributions:

1. For $n$-player games we improve the bound on the size of the support of an approximate Nash equilibrium from $O(n^2 \log m)$ (see Lipton et. al. \cite{lipton}) to $O(n \log m)$ (see Corollary \ref{cor:gen}).

2. We generalize the result of Daskalakis and Papadimitriou \cite{daskalakis} from two-player games case to all $n$-player game cases (see Corollary \ref{cor:2}). 

3. We provide a logarithmic bound ($O(\log n + \log m)$) on the size of the support of approximate Nash equilibrium in the case of graphical games. This bound is novel (see Theorem \ref{theo:1}).

The note is organized as follows. In Section \ref{sec:pre} we present the notations and preliminaries that will be useful in our new proof. In Section \ref{sec:gen} we state and prove the a result on graphical games; this result generalizes Lipton et al. \cite{lipton}. In Section \ref{sec:small} we state and prove the result that generalizes the result of Daskalakis and Papadimitriou. Section \ref{sec:disc} is a discussion.

\section{Preliminaries}\label{sec:pre}

We consider $n$-player games where every player $i$ has a large number of actions. For simplicity, we will consider the case where all players have the same number of actions\footnote{Given a game where player $i$ has $m_i$ actions, we can consider an equivalent game where every player has $m=\max_i m_i$ actions. This can be done by adding $m-m^i$ strictly dominated actions to every player $i$.} $m$. We will use the following standard notations. We denote by $A_i=\{1,2,...,m\}$ the \emph{actions set of player} $i$, and by $A=\times_i A_i$ the \emph{actions profile set}. The simplex $\Delta(A_i)$ is the set of \emph{mixed strategies of player} $i$. We will assume that the payoffs of all players are in $[0,1]$, and $u_i:A\rightarrow [0,1]$ will denote the \emph{payoff function of player} $i$. The payoff function $u_i$ can be multylinearly extended to $u_i:\Delta(A)\rightarrow [0,1]$. The \emph{payoff functions profile} is $u=(u_i)_{i=1}^n$, which is also called the \emph{game}. A mixed action profile $x=(x_1,x_2,...,x_n)$ is an \emph{Nash} $\varepsilon$\emph{-equilibrium} if for every action $a_i \in A_i$, it holds that $u_i(x)\geq u_i(a_i,x_{-i})-\varepsilon$. 

A mixed strategy $x_i=(x_i(1),x_i(2),...,x_i(m))$ of player $i$ will be called $k$\emph{-uniform} if $x_i(j)=c_j/k$, where $c_j\in \mathbb{N}$ for every $j=1,2,...,m$. Note that the support of $k$ uniform strategy is of size at most $k$. A mixed strategy profile $x=(x_i)_{i=1}^n$ will be called $k$\emph{-uniform} if every $x_i$ is $k$-uniform.

We say that the payoff of player $i$ \emph{depends on player} $j$ if there exists an action profile $a_{-j}$ and a pair of actions $a_j,a'_j$ of player $j$ such that $u_i(a_j,a_{-j}) \neq u_i(a'_j,a_{-j})$. A game where the payoff of every player depends on at most $d$ other players will be called a \emph{graphical game of degree} $d$. Graphical games, introduced by Kearns et al. \cite{kea}, express the situation where players are located on vertices of an underlying graph and their payoffs are influenced only by their neighbors' actions. Note that every $n$-player game is a graphical game of degree $n-1$.

\subsection{Lipschitz games}

Player $i$ has a $\lambda$\emph{-Lipschitz payoff function} if $|u_i(a_j,a_{-j})-u_i(a'_j,a_{-j})|\leq \lambda$ for every $i\neq j$ and every $a_j,a'_j\in A_j$.  The Lipschitz property means that a change of strategy of a single player $j \neq i$ has little effect on the payoff of player $i$. Note that player $i$ can have a big effect on his own payoff. A game will be called $\lambda$\emph{-Lipschitz} if the payoff functions of all players are $\lambda$-Lipschitz.

The following proposition is an important property of $\lambda$-Lipschitz games.

\begin{proposition}\label{prop:lip}
If in an $n$-player game the payoff of player $i$ depends on at most $d$ players, and his payoff function is $\lambda$-Lipschitz, then for every pure action $a_i\in A_i$ and for every mixed action profile of the opponents $x_{-i}$, it holds that\footnote{By the notation $x_{-i}(B)$, we refer to $x_{-i}$ as a probability measure on $A_{-i}$, and so $x_{-i}(B)$ is the probability of the event $B$.} 
\begin{equation*}
x_{-i}(B)\geq 1-2\exp \left( -\frac{\delta^2}{d\lambda^2} \right)
\end{equation*}
where $B \subset A_{-i}$ is defined by
\begin{equation*}
B=\{ a_{-i}: |u_i(a_i,a_{-i})-u_i(a_i,x_{-i})|\leq \delta \}.
\end{equation*}
\end{proposition}
In simple words, Proposition \ref{prop:lip} claims that if we randomize an action profile $a_{-i}$ according to $x_{-i}$, then probably player $i$ will have approximately the same outcome if he plays against $a_{-i}$ or against $x_{-i}$.

Proposition \ref{prop:lip} is based on the concentration of measure phenomena for Lipschitz functions (see Ledoux \cite{ledo}) and it is derived explicitly in Azrieli and Shmaya \cite{shmaya}.

\subsection{From general games to Lipschitz games}\label{sec:G->L}

We present a very natural procedure that constructs for every game a corresponding game with the Lipschitz property.

Fix $k\in \mathbb{N}$. Given a game $u$ we construct a new game $v=v(u,k)$ with $kn$ players as follows. We ``split'' every player $i$ into a population of $k$ players $i(1),i(2),...,i(k)$. Each player $i(j)$ plays the original game $u$ against the aggregate behavior of the $n-1$ other populations of size $k$.

Formally, it will be convenient to present $A_i$ as the set of vectors \newline $\{e_1,e_2,...,e_m \}\subset \mathbb{R}^m$, where $e_j$ is the $j-th$ unit vector in $\mathbb{R}^m$. In such a representation the unit simplex $\Delta^m :=\{(x_j)_{j=1}^m: \sum_j =1, x_j\geq 0\}$ is the set of mixed strategies $\Delta(A_i)$. All players $i(j)$ have the same actions set $A_i$.  The payoff of player $i_0(j_0)$ is defined by
\begin{equation*}
v_{i_0(j_0)}((a_{i(j)})_{1\leq i \leq n, 1 \leq j \leq k})=u_i \left( a_{i_0(j_0)},\left( \frac{\sum_{j=1}^k a_{i(j)}}{k}  \right)_{i \neq i_0} \right).
\end{equation*}
Note that $\sum_{j=1}^k a_{i(j)}/k \in \Delta^m$; therefore, this vector represents the mixed strategy of population $i$.

\begin{remark}\label{rem:lip}
The game $v$ has the following two properties: 

(P1) $v$ is $1/k$ Lipschitz, because a deviation of a single player $i(j)$ changes the mixed strategy that is played by population $i$ only by $1/k$.

(P2) Every pure Nash $\varepsilon$-equilibrium of the game $v$ corresponds to a $k$-uniform mixed Nash $\varepsilon$-equilibrium of the game $u$. The corresponding mixed equilibrium will be the one where player $i$ plays the aggregated strategy of population $i$ in the game $v$.\footnote{Moreover, the opposite direction is also true. Every $k$-uniform $\varepsilon$-equilibrium of $u$ corresponds to a pure Nash $\varepsilon$-equilibrium of $v$. The corresponding pure equilibrium will be the one where population $i$ plays a pure profile with aggregated behavior $x_i$, where $x_i$ is the $k$-uniform strategy of player $i$ in the $\varepsilon$-equilibrium in the game $u$.} 
\end{remark}

\section{General Games and Graphical Games}\label{sec:gen}

\begin{theorem}\label{theo:1}
Every $n$-player graphical game of degree $d$ with $m$ actions for every player has a $k$-uniform Nash $\varepsilon$-equilibrium for $k=\frac{8}{\varepsilon^2}d(\log n+\log m)$.
\end{theorem}

Usually graphical game models consider games with a large number of players $n$ of constant degree $d$. Theorem \ref{theo:1} proves the existence of a relatively simple approximate  Nash equilibrium where every player uses a strategy with a support that is logarithmically small on $n$ and $m$.

Lipton et al. \cite{lipton} show that in every $n$-player game with $m$ actions for every player there exists a $k$-uniform Nash $\varepsilon$-equilibrium for $k=O(n^2 \log m)$. Theorem \ref{theo:1} applied to general games shows that in such games there exists a $k$-uniform Nash $\varepsilon$-equilibrium for $k=O(n \log m)$.

\begin{corollary}\label{cor:gen}
Every $n$-player game of with $m$ actions for every player has a $k$-uniform Nash $\varepsilon$-equilibrium for $k=\frac{8}{\varepsilon^2}(n-1)(\log n+\log m)$.
\end{corollary}

As a straightforward corollary of this result, we derive the following improvement to the oblivious algorithm for computing Nash approximate equilibrium in games with $n$ players. 

\begin{corollary}
Let $k=\frac{8}{\varepsilon^2}(n-1)(\log n+\log m)$. Then the oblivious algorithm\footnote{The term ``oblivious algorithm'' is from \cite{daskalakis}.} that exhaustively searches over the $k$-uniform strategies finds an $\varepsilon$-equilibrium in $O(m^{n^2 \log m })$ steps.\footnote{Lipton et al. \cite{lipton} prove a bound of $O(m^{n^3 \log m })$ on the number of steps.}
\end{corollary}

\begin{proof}[Proof of Theorem \ref{theo:1}]

Let $k=\frac{8}{\varepsilon^2}d(\log n+\log m)$. We construct the game $v=v(u,k)$ as presented in Section \ref{sec:G->L}. We prove that the game $v$ possesses a pure Nash equilibrium, then, by Remark \ref{rem:lip} (P2) this concludes the proof. Moreover, we will prove that every $nk$-player $1/k$-Lipschitz graphical game of degree $dk$ has a pure Nash $\varepsilon$-equilibrium.

Consider a mixed action profile $x$ that is a (possibly mixed) Nash equilibrium of $v$. For every player $i$ and every action $b\in A_i$ of player $i$, we define the set of action profiles 
\begin{equation*}
E_{i,b}:=A_i\times \{a_{-i}: |v_i(b,a_{-i})-v_i(b,x_{-i})|\leq \varepsilon/2\}\subset A. 
\end{equation*}

Every action $a^*\in \cap_{i,b} E_{i,b} \cap support(x)$ is a pure Nash $\varepsilon$-equilibrium according to the following inequality:
\begin{equation*}
v_i(d,a^*_{-i})\leq v_i(d,x_{-i}) + \frac{\varepsilon}{2} \leq v_i(a^*_i,x_{-i}) + \frac{\varepsilon}{2} \leq v_i(a^*_i,a^*_{-i})+\varepsilon,
\end{equation*}
where the first inequality follows from $a^* \in E_{i,d}$, the second from $a^*_i\in support(x^i)$, and the third from $a^* \in E_{i,a^*_i}$. Therefore it is enough to prove that the above intersection is not empty.

By proposition \ref{prop:lip} we have
\begin{equation}\label{ineq1}
x(E_{i,b}^c)\leq 2 \exp(-\frac{\varepsilon^2 k}{4d}).
\end{equation}
Putting $k=\frac{8}{\varepsilon^2}d(\log n+\log m)$ we get $x(E_{i,b}^c)\leq 1/(2nkm)$. There are $nk$ players in $v$, and $m$ actions for every player. Therefore there are $nkm$ events $E_{i,b}$. Therefore, $x(\cap E_{i,b})\geq 1/2 >0$, which concludes the proof. 
\end{proof}

\section{Small Probability Games}\label{sec:small}

Following the terminology of \cite{daskalakis}, a profile of mixed actions $x$ will be called a $c$\emph{-small probabilities profile} if $x_i(j)\leq c/m$ for every player $i$ and every $j\in A^i$. A game $u$ will be called a $c$\emph{-small probability game} if there exists a Nash equilibrium $x$ that is a $c$-small probability profile.

Daskalakis and Papadimitriou \cite{daskalakis} prove that in small probability two-player games the oblivious random algorithm that samples $k$-uniform strategies for $k=\Theta(log m)$ finds an approximate Nash equilibrium in $O(c^2 m^{\log c})$ steps, i.e., in polynomial time in $m$. Here we generalize this result to general $n$-player games.

It will be convenient to think of the $k$-uniform strategies as a multiset that contains $k$ \emph{ordered} actions. In such a case the set of $k$-uniform strategy profiles is of size $m^{kn}$.

\begin{theorem}\label{theo:2}
Let $u$ be an $n$-player $c$-small probability games with $m$ actions for every player, and let $k=\frac{8}{\varepsilon^2}(n-1)(\log n+\log m)$. Then, among the $m^{kn}$ $k$-uniform strategy profiles in $u$, the number of strategy profiles that forms an Nash $\varepsilon$-equilibrium is at least
\begin{equation*}
\frac{m^{kn}}{2(nm)^{\frac{8}{\varepsilon^2}(n-1)n\ln c}}.
\end{equation*}
\end{theorem}

\begin{corollary}\label{cor:2}
Fix $n$ and let $k=\frac{8}{\varepsilon^2}(n-1)(\log n+\log m)$. Then the oblivious algorithm that samples at random $k$-uniform strategies and checks whether it forms an $\varepsilon$-equilibrium finds such an $\varepsilon$-equilibrium in $c$-small probability games after $(nm)^{\frac{8}{\varepsilon^2}(n-1)n\ln c}$ samples in expectation, i.e., after polynomial time in $m$.
\end{corollary}

\begin{proof}[Proof of Theorem \ref{theo:2}]
Fix $k=\frac{8}{\varepsilon^2}(n-1)(\log n+\log m)$, and let $x$ be a $c$-small probability equilibrium of $u$. Consider the game $v=v(u,k)$ that is defined in Section \ref{sec:G->L}. Note that the action profile where every player $i(j)$ plays the mixed action $x_i$ is a Nash equilibrium of the game $v$. Denote this equilibrium by $x^v$.

Following the same analysis that was done in the proof of Theorem \ref{theo:1} we define the sets $E_{i,b}$ and we know that $x^v(\cap_{i,b} E_{i,b})\geq 1/2$. Two different pure action profiles in $\cap_{i,b} E_{i,b}\cap support(x^v)$ correspond to two different $k$-uniform Nash $\varepsilon$-equilibria in $u$. Let us show that there are many different action profiles in $\cap_{i,b} E_{i,b}\cap support(x^v)$.

Note that $x^v(a)\leq (c/m)^{nk}$ because $x$ is a $c$-small probabilities profile. On the other hand, $x^v(\cap_{i,b} E_{i,b})\geq 1/2$. Therefore, there must be at least $m^{nk}/2c^{nk}$ different profiles in $\cap_{i,b} E_{i,b}\cap support(x^v)$, which yield that there are at least $m^{nk}/2c^{nk}$ different $k$-uniform Nash $\varepsilon$-equilibria in $u$.

It only remains to evaluate the expression $c^{nk}$:
\begin{equation*}
c^{nk}=\left( c^{\ln n+\ln m}\right)^{\frac{8}{\varepsilon^2}(n-1)n}= \left( n^{\ln c} m^{\ln c} \right)^{\frac{8}{\varepsilon^2}(n-1)n}=(nm)^{\frac{17}{\varepsilon^2}(n-1)n\ln c}.
\end{equation*}

\end{proof}

\section{Discussion}\label{sec:disc}

This note contains a new approach to the problem of an approximate small support Nash equilibrium. Instead of considering the game itself, we can consider a population game where every player is replaced by a population of players and analyze the existence of an approximate pure Nash equilibrium in the population game. I believe that this approach might be useful for analyzing other interesting questions. For example, the question of characterizing the class of two-player games where an approximate Nash equilibrium with constant support exists might have the following interpretation: which two-population games with constant population size has a pure Nash equilibrium? Clearly, characterization of the above class is an important question because for those games there exists a polynomial-time exhaustive search algorithm for computing an approximate Nash equilibrium.

This paper provides an upper bound of $O(n \log m)$ on the size of the support of an approximate Nash equilibrium. It is known that the bound $\log m$ is tight even in two-player games (see Althofer \cite{althofer}); i.e., there exists a two-player game where no Nash approximate equilibrium with a support smaller than $c\log m$ exists. The question whether the linear dependence on $n$ is also tight remains an open question.

\textbf{Open problem:} Does there exist an $n$-player $n$-action game where in every Nash approximate equilibrium at least one of the players plays a mixed action with support of size $f(n)$?

By Althofer \cite{althofer} the answer to this question for $f(n)=c\log n$ is positive. What about $f(n)=c n^\alpha$ for $\alpha< 1$? What about $f(n)=cn$?


\begin{thebibliography}{99}

\bibitem{althofer} Althofer, I. (1994) ``On Sparse Approximations to Randomized Strategies and Convex Combinations,'' \emph{Linear Algebra and its Applications} 199,  339–-355.

\bibitem{shmaya} Azrieli, Y. and Shmaya, E. (2013) ``Lipschitz Games,'' \emph{Mathematics of Operations Research}, forthcoming.

\bibitem{lipton} Lipton, R. J., Markakis, E., and Mehta, A. (2003) ``Playing Large Games Using Simple Strategies,'' \emph{Proceedings of the 4th ACM Conference on Electronic Commerce} pp. 36--41.  

\bibitem{daskalakis} Daskalakis, C. and Papadimitriou, C. H. (2009) ``On Oblivious PTAS's for Nash Equilibrium,'' \emph{Proceedings of the 41st Annual ACM Symposium on Theory of Computing}, pp. 75--84.


\bibitem{kea} Kearns, M. J., Littman, M. L., and  Singh, S. P. (2001) ``Graphical Models for Game Theory,'' \emph{Proceedings of the 17th Conference on Uncertainty in Artificial Intelligence}, pp. 253--260.

\bibitem{ledo} Ledoux, M. (2001) \emph{The Concentration of Measure Phenomenon,} American Mathematical Society, Providence, RI.

\bibitem{rosenthal} Rosenthal, R. W. (1973). ``A Class of Games Possessing Pure-Strategy Nash Equilibria,'' \emph{International Journal of Game Theory} 2, 65--67.

\bibitem{shmid} Shmeidler, D. (1973). ``Equilibrium Points of Nonatomic Games,'' \emph{Journal of Statistical Physics} 7, 295--300.

\end{thebibliography}
\end{document}